\documentclass[12pt,a4paper]{article}

\usepackage{amssymb}
\usepackage{amsmath, amsthm}
\usepackage{arydshln}
\usepackage{epsfig}
\usepackage{setspace}

\usepackage{geometry}

\def\ZZ{\mathbb{Z}}
\def\QQ{\mathbb{Q}}

\numberwithin{equation}{section}

\newcommand{\rank}{\mathop{\rm rank} }

\newcommand{\Det}{\mathop{\rm Det} }

\newtheorem{Thm}{Theorem}[section]
\newtheorem{Prop}[Thm]{Proposition}
\newtheorem{Lem}[Thm]{Lemma}

\theoremstyle{definition}

\newtheorem{Ex}[Thm]{Example}

\title{On a weighted linear matroid intersection algorithm by 
	deg-det computation}
\author{Hiroki FURUE and Hiroshi HIRAI \\
Department of Mathematical Informatics, \\
Graduate School of Information Science and Technology,   \\
The University of Tokyo, Tokyo, 113-8656, Japan.\\
\texttt{\normalsize hiroki\_furue@mist.i.u-tokyo.ac.jp}\\
\texttt{\normalsize hirai@mist.i.u-tokyo.ac.jp}
}
\begin{document}
	\maketitle
	\begin{abstract}
		In this paper, we address the weighted linear matroid intersection problem
		from computation of the degree of the determinant of a symbolic matrix.
		We show that a generic algorithm computing 
		the degree of noncommutative determinants,  
		proposed by the second author, 
		becomes an $O(mn^3 \log n)$ time algorithm 
		for the weighted linear matroid intersection problem, 
		where two matroids are given by column vectors of $n \times m$ matrices $A,B$.
	    We reveal that our algorithm is viewed as a ``nonstandard" implementation 
	    of Frank's weight splitting algorithm for linear matroids. 
	    This gives a linear algebraic reasoning to Frank's algorithm.
	    Although our algorithm is slower than existing algorithms in the worst case estimate, 
	    it has a notable feature. Contrary to existing algorithms,
	    our algorithm works on different matroids 
	    represented by another ``sparse" matrices $A^0,B^0$, 
	    which skips unnecessary Gaussian eliminations for constructing residual graphs.
	\end{abstract}


\section{Introduction}
Several basic combinatorial optimization problems 
have linear algebraic formulations.
It is classically known \cite{Edmonds67} that the maximum cardinality of a matching 
in a bipartite graph $G = (U,V;E)$ with color classes $U = [n],V = [n']$
is equal to the rank of the matrix $A = \sum_{e \in E} A_{e} x_{e}$, 
where $x_{e}$ $(e \in E)$ are variables and 
$A_{e}$ is an $n \times n'$ matrix 
with $(A_{e})_{ij} := 1$ if $e = ij$ and zero otherwise. 
Such a rank interpretation is known for the linear matroid intersection,
nonbipartite matching, and linear matroid matching problems; see \cite{Lovasz89}.

The degree of the determinant of a polynomial (or rational) matrix 
is a weighted counter part of rank, 
and can formulate weighted versions of combinatorial optimization problems.
The maximum weight perfect matching problem 
in a bipartite graph $G = ([n],[n];E)$ with integer weights $c_e$ $(e \in E)$
corresponds to computing the degree $\deg_t \det A(t)$ of the determinant 
of the (rational) matrix $A(t) := \sum_{e \in E} A_{e} x_{e}t^{c_e}$.
Again, the weighted linear matroid intersection,
nonbipartite matching, and linear matroid matching problems
admit such formulations.

Inspired by the recent advance~\cite{GGOW15,IQS15b} of a noncommutative approach 
to symbolic rank computation, 
the second author~\cite{HH_degdet} 
introduced the problem of 
computing the degree $\deg_t \Det A(t)$ of 
the {\em Dieudonn\'e determinant} $\Det A(t)$ of a matrix $A(t) = \sum_{i} A_i(t) x_i$, where 
$x_i$ are pairwise noncommutative variables and
$A_i(t)$ is a rational matrix with commuting variable $t$. 
He established a general min-max formula for $\deg_t \Det A(t)$, 
presented a conceptually simple and generic algorithm, referred here to as ${\bf Deg\mbox{-}Det}$, for computing $\deg_t \Det A(t)$, 
and showed that $\deg_t \det A(t) = \deg_t \Det A(t)$ holds
if $A(t)$ corresponds to an instance of the weighted linear matroid intersection problem.
In particular, ${\bf Deg\mbox{-}Det}$ gives rise to a pseudo-polynomial time algorithm for 
the weighted linear matroid intersection problem.
In the first version of the paper~\cite{HH_degdet}, 
the second author asked 
(i) whether ${\bf Deg\mbox{-}Det}$ can be a (strongly) polynomial time algorithm for 
the weighted linear matroid intersection, 
and (ii) how ${\bf Deg\mbox{-}Det}$ is related 
to the existing algorithms for this problem. 
He pointed out some connection of ${\bf Deg\mbox{-}Det}$ 
to the primal-dual algorithm by Lawler~\cite{Lawler75} 
but the precise relation was not clear.

The main contribution of this paper is to answer the questions (i) and (ii):
\begin{itemize}
	\item We show that ${\bf Deg\mbox{-}Det}$ becomes an $O(nm^3 \log n)$ time algorithm for the weighted linear matroid intersection problem, 
	where the two matroids are represented and given by two $n \times m$ matrices $A,B$. 
	This answers affirmatively the first question.
	\item For the second question, we reveal the relation between our algorithm and 
	the {\em weight splitting algorithm} by Frank~\cite{Frank81}.
	This gives a linear algebraic reasoning to Frank's algorithm.
     We show that the behavior of our algorithm is precisely the same 
     as that of a slightly modified version of Frank's algorithm. 
     However our algorithm is rather different from the standard implementation 
     of Frank's algorithm for linear matroids.
     This relationship was unexpected and nontrivial for us, 
     since the two algorithms look quite different.
\end{itemize}

Although our algorithm is slower than the standard $O(mn^3)$-time implementation 
of Frank's algorithm in the worst case estimate, it has a notable feature.  
Frank's algorithm works on a subgraph $\bar G_X$ of the residual graph $G_X$ 
for a common independent set $X$,
where $G_X$ is determined by Gaussian elimination for $A,B$ and 
$\bar G_X$ is determined by a splitting of the weight.
On the other hand, our algorithm does not compute the residual graph $G_X$ but computes  
a non-redundant subgraph $G_X^0$ of $\bar G_X$, 
which is the residual graph of different matroids 
represented by another ``sparse" matrices $A^0,B^0$.
Consequently, our algorithm applies fewer elimination operations 
than the standard one, which will be a practical advantage.

\paragraph{Related work.}
The essence of {\bf Deg-Det} comes from the {\em combinatorial relaxation algorithm} 
by Murota~\cite{Murota95_SICOMP}, 
which is an algorithm computing the degree of 
the (ordinary) determinant of a polynomial/rational matrix; see~\cite[Section 7.1]{MurotaMatrix}.

Several algorithms have been proposed for the general weighted matroid intersection 
problem under the independence oracle model; see e.g., \cite[Section 41.3]{Schrijver} and the references therein. 
For linear matroids given by two $n \times m$ matrices, 
the current fastest algorithms (as far as we know)
are an $O(mn^\omega)$-time implementation of Frank's algorithm using fast matrix multiplication
and an $O(n m^{\frac{7-\omega}{5-\omega}} \log^{\frac{\omega-1}{5- \omega}} n \log m C)$-time algorithm by Gabow and Xu~\cite{Gabow96}, 
where $C$ is the maximum absolute value of weights $c_i$
and $\omega \in \left[ 2, 2.37 \right]$ denotes the exponent of the time complexity of matrix multiplication.
Huang, Kakimura, and Kamiyama~\cite{HKK2019} 
gave an $O(n m \log n_* + C m n _* ^{\omega -1})$-time algorithm, where $n_*$ is the maximum size of a common independent set.
This algorithm is currently fastest for the case of small $C$.

For unweighted linear matroid intersection,
Cunningham~\cite{Cunningham} showed that the classical Edmonds' algorithm
runs in $O(mn^2 \log n)$ time.
Harvey~\cite{Harvey09} gave a randomized $O(mn^{\omega-1})$-time algorithm.
His algorithm also treats the problem as the rank computation
of a matrix with variables $x_i$, and
uses random substitution of the variables
and fast matrix multiplication.

\paragraph{Organization.}
The rest of this paper is organized as follows. 
In Section~\ref{sec:preliminaries}, 
we introduce algorithm ${\bf Deg\mbox{-}Det}$, 
and describe basics of
the unweighted (linear) matroid intersection problem 
from a linear algebraic viewpoint; 
our algorithm treats the unweighed problem as a subproblem. 
In Section~\ref{sec:algorithm}, 
we first formulate the weighted linear matroid intersection problem as 
the degree of the determinant of a rational matrix $A$, and show 
that ${\bf Deg\mbox{-}Det}$ computes $\deg_t \det A$ correctly.  
Then we present our algorithm by specializing  ${\bf Deg\mbox{-}Det}$, 
analyze its time complexity, and reveal its relationship to Frank's algorithm.

In this paper, we deal with linear matroids represented over the field of rationals but our augment and algorithm work on an arbitrary field.

\section{Preliminaries}\label{sec:preliminaries}

\subsection{Notation}
Let $\QQ$ and $\ZZ$ denote the sets of rationals and integers, respectively.
Let ${\bf 0} \in \QQ^n$ denote the zero vector.
For $I \subseteq [n]:=\{1,2,...,n\}$, let ${\bf 1}_I  \in {\mathbb Q}^n$ denote the characteristic vector of $I$, that is, $({\bf 1}_{I})_k := 1$ if $k \in I$ and $0$ otherwise.
Here, ${\bf 1}_{[n]}$ is simply denoted by ${\bf 1}$.

For a polynomial $p = \sum_{i=0}^k a_{i} t^i \in \QQ[t]$ with $a_k \neq 0$, 
the degree $\deg_t p$ with respect to $t$
is defined as $k$.
The degree $\deg_t p/q$ of a rational function 
$p/q \in \QQ(t)$ with polynomials $p,q \in \QQ[t]$ 
is defined as $\deg_t p - \deg_t q$.
The degree of zero polynomial is defined as $- \infty$.

A rational function $p/q$ is called {\em proper} if $\deg_t p/q \leq 0$.
A rational matrix $Q \in \QQ(t)^{n \times m}$ is called proper 
if each entry of $Q$ is proper.
For a proper rational matrix $Q \in \QQ(t)^{n \times m}$, 
there is a unique matrix over $\QQ$, denoted by $Q^0$, such that
\[
Q = Q^0 + t^{-1} Q',
\]
where $Q'$ is some proper matrix.

For an integer vector $\alpha \in \ZZ^n$, let $(t^{\alpha})$ denote 
the $n \times n$ diagonal matrix having diagonals $t^{\alpha_1}, t^{\alpha_2},\ldots,t^{\alpha_n}$ in order, that is,
\[
(t^{\alpha}) =
\left( \begin{array}{cccc}
t^{\alpha_1} &                    &  & \\
                 & t^{\alpha_2}   &  & \\
                 &                     & \ddots & \\
                 &                      &           & t^{\alpha_n} 
\end{array}
\right).
\]

For a matrix $A \in \QQ^{n \times m}$ and $J \subseteq [m]$,
let  $A[J]$ denote the submatrix of $A$ consisting of 
the $j$-th columns for $j \in J$.
Additionally, for $I \subseteq [n]$, 
let $A[I,J]$ denote the submatrix of $A$ 
consisting of the $(i,j)$-entries for $i \in I,j \in J$.

\subsection{Algorithm {\bf Deg-Det}}\label{subsec:Deg-Det}
Given $n \times n$ rational matrices $M_1,M_2,\ldots,M_m \in \QQ(t)^{n \times n}$, 
consider the following matrix
\begin{equation*}
M := M_1 x_1 + M_2 x_2 + \cdots + M_m x_m \quad \in \QQ(t, x_1,x_2,\ldots,x_m), 
\end{equation*}
where $x_1,x_2,\ldots,x_m$ are variables and $M$ is regarded 
as a multivariate rational matrix with (pairwise commutative) 
variables $t, x_1,x_2,\ldots,x_m$.
We address the computation of the degree of the determinant of $M$ with respect to $t$.

Consider the following optimization problem:
\begin{eqnarray*}
{\rm (P)} \quad \mbox{Max.} &&  \deg_t \det P + \deg_t \det Q \\
\mbox{s.t.} && PM Q: \mbox{proper,}\\
&& P,Q \in \QQ(t)^{n \times n}: \mbox{nonsingular.}    
\end{eqnarray*}
This problem gives an upper bound of $\deg_t \det M$.
Indeed, if $P M Q$ is proper, 
then $\deg_t \det PMQ \leq 0$, 
and  $\deg_t \det M \leq - \deg_t \det P - \deg_t \det Q$.
In fact, it is shown \cite{HH_degdet} that the optimal value of (P)
is interpreted as the negative of the degree of the {\em Dieudonn\'e determinant} 
of $M$ for the case where $x_1,x_2,\ldots,x_m$ are pairwise noncommutative variables.

The following algorithm for (P) is due to \cite{HH_degdet}, 
which is viewed as a simplification of the {\em combinatorial relaxation algorithm} 
by Murota~\cite{Murota95_SICOMP}; see also \cite[Section 7.1]{MurotaMatrix}.
\begin{description}
	\item[Algorithm: Deg-Det]
	\item[Input:] $M = M_1x_1 + M_2 x_2 + \cdots + M_m x_m$, where $M_i \in \QQ(t)^{n \times n}$ 
	for $i \in [m]$. 
	\item[Output:] An upper bound of $\deg_t \det M$ (the negative of the optimal value of (P)).
	\item[0:]  Let $P := t^{-d} I$ and $Q := I$, where $d$ is the maximum degree of entries 
	in $M$.  Let $D^* := nd$.
	\item[1:] Solve the following problem:
	\begin{eqnarray*}
		({\rm P}^0) \quad {\rm Max.} && r + s  \\
		{\rm s.t.} && \mbox{$K (PMQ)^0 L$ has an $r \times s$ zero submatrix,} \\
		&& K, L \in \QQ^{n \times n}: \mbox{nonsingular},
	\end{eqnarray*}
	and obtain optimal matrices $K,L$; recall the notation $(\cdot)^{0}$ in Section 2.1.
	\item[2:] If the optimal value $r+s$ is at most $n$, 
	then stop and output $D^*$.
	\item[3:] 
	Let $I$ and $J$ be the sets of row and column indices, respectively, of 
	the $r \times s$ zero submatrix of $K (PMQ)^0 L$.
	Find the maximum integer $\kappa (\geq 1)$ 
	such that
	$(t^{\kappa{\bf 1}_{I}})K PMQ L 
	(t^{- \kappa {\bf 1}_{[n] \setminus J}})$ is proper. 
	
	If $\kappa$ is unbounded, then output $-\infty$. 
	Otherwise, let $P \leftarrow (t^{\kappa{\bf 1}_{I}})KP$, 
	$Q \leftarrow QL(t^{-\kappa {\bf 1}_{[n]\setminus  {J}}})$ and  
	$D^* \leftarrow D^* - \kappa (r+s -n)$.
	Go to step 1.
\end{description}
Observe that  in each iteration $(P,Q)$ is a feasible solution of (P), and
$D^*$ equals $- \deg_t \det P - \deg_t \det Q$.
Thus, (P) gives an upper bound of $\deg_t \det M$.
We are interested in the case where the algorithm outputs $\deg_t \det M$ correctly.
\begin{Lem}[\cite{HH_degdet}]\label{lem:optimality}
	In step 2 of {\bf Deg-Det}, the following holds: 
	\begin{itemize}		
		\item[{\rm (1)}] If $r+s > n$, then $(PMQ)^0$ is singular over $\QQ(x_1,x_2,\ldots,x_m)$.
		\item[{\rm (2)}] If $(PMQ)^0$ is nonsingular, 
		then $D^* = \deg_t \det M$.
	\end{itemize}
\end{Lem}
\begin{proof}
	(1). It is obvious that
	any $n \times n$ matrix is singular if it has an 
	$r \times s$ zero submatrix with $r+s > n$.
	
	(2).
	$P MQ$ is written as $(PMQ)^0 + t^{-1} N$ for some proper $N$.
	If $(PMQ)^0$ is nonsingular, 
	then $\deg_t \det PMQ = \deg_t \det (PMQ)^0 = 0$, 
	and hence  $\deg_t \det M = - \deg_t \det P - \deg_t \det Q = D^*$.
\end{proof}
	
\subsection{Algebraic formulation for linear matroid intersection}\label{subsec:formulation}
Let $A = (a_1\ a_2\ \cdots\ a_m)$ be an $n \times m$ matrix over $\QQ$.
Let ${\bf M}(A) = ([m], {\cal I}(A))$ denote the linear matroid represented 
by $A$. Specifically, the ground set of the matroid ${\bf M}(A)$ is 
the set $[m]$ of the column indices, and the family ${\cal I}(A)$ 
of independent sets of ${\bf M}(A)$ consists of all subsets $X \subseteq [m]$
such that the corresponding column vectors $a_i$ $(i \in X)$ are linearly independent.
Let $\rho_A: 2^{[m]} \to \ZZ$ 
denote the rank function of ${\bf M}(A)$, that is,  
$\rho_A(X) := \max \{ |Y| \mid Y \in {\cal I}(A), Y \subseteq X\}$.
A minimal (linearly) dependent subset is called a circuit.
See, e.g., \cite[Chapter39]{Schrijver} for basics on matroids.

Suppose that we are given another $n \times m$ matrix 
$B = (b_1\ b_2\ \cdots\ b_m) \in \QQ^{n \times m}$.
Let ${\bf M}(B) = ([m], {\cal I}(B))$ be the corresponding linear matroid.
A common independent set of ${\bf M}(A)$ and ${\bf M}(B)$ is a subset $X \subseteq [m]$ such that $X$ 
is  independent for both ${\bf M}(A)$ and ${\bf M}(B)$.
The linear matroid intersection problem is to find a common independent set of the maximum cardinality. 
To formulate this problem linear algebraically, 
define an $n \times n$ matrix $M = M(A,B)$ over $\QQ(x_1,x_2,\ldots,x_m)$ by
\[
M := \sum_{i=1}^m a_i b_i^{\top} x_i,
\]
where $x_1,x_2,\ldots,x_m$ are variables.
The following is the matroid intersection theorem and its linear algebraic sharpening.
\begin{Thm}[{\cite{Edmonds79}; see also \cite{Lovasz89,TI74}}]\label{thm:MatroidIntersection}
The following quantities are equal:
\begin{itemize}
	\item[{\rm (1)}] The maximum cardinality of a common independent set of ${\bf M}(A)$ and ${\bf M}(B)$.
	\item[{\rm (2)}] The minimum of $\rho_A(J) + \rho_B([m] \setminus J)$ 
	over $J \subseteq [m]$.
	\item[{\rm (1$'$)}] $\rank M$.
	\item[{\rm (2$'$)}] $2n$ minus the maximum of $r+s$ 
	such that 
	$K M L$ has an $r \times s$ zero submatrix for some nonsingular matrices $K,L \in \QQ^{n \times n}$.
\end{itemize}	
\end{Thm}
\begin{proof}[Sketch of Proof]
	(1) $=$ (2) is nothing but the matroid intersection theorem.
	
	(1) $=$ (1$'$). A $k \times k$ submatrix $M'$ of $M$ is 
	represented by $M' = A' D B'^{\top}$, where $A',B'$ 
	are $k \times m$ submatrices of $A,B$, 
	and $D$ is the diagonal matrix with diagonals $x_1,x_2,\ldots,x_m$ (in order). 
	From Binet-Cauchy formula, we see that $\det M' \neq 0$ if and only 
	if there is a $k$-element subset $X \subseteq [m]$ such that $\det A'[X] \det B'[X] \neq 0$. 
	Thus, $\rank M \geq k$ if and only if 
	there is a common independent set of cardinality $k$.

	(2) $\geq$ (2$'$). Take a basis $u_1,u_2,\ldots,u_r$ 
	of the orthogonal complement of the vector space spanned by $\{a_i \mid i \in J\}$, and extend it to a basis $u_1,u_2,\ldots,u_n$ of $\QQ^n$, where $r = n - \rho_A(J)$.
	Similarly, take a basis $v_1,v_2,\ldots,v_n$ of $\QQ^n$ that 
	contains a basis $v_1,v_2,\ldots,v_s$ of  the orthogonal complement of the vector space spanned by $\{b_i \mid i \in [m] \setminus J\}$, where $s = n - \rho_{B}([m] \setminus J)$.
   Then $u_k^{\top} a_i b_i^{\top} v_{\ell} = 0$ for all $k\in [r]$, 
   $\ell \in [s]$, and $i \in [m]$.
    This means that $KML$ has an $r \times s$ zero submatrix 
	for $K = (u_1\ u_2\ \cdots\ u_n)^{\top}$ and $L = (v_1\ v_2\ \cdots\ v_n)$.
	
	(2$'$) $\geq$ (1$'$). If $KML$ has an $r \times s$ zero submatrix,
	then $\rank M = \rank KML \leq n - r + n- s$.
\end{proof}

Let us briefly explain Edmonds' algorithm to obtain 
a common independent set of the maximum cardinality.
For any common independent set $X$, 
the auxiliary (di)graph $G_X = G_X(A,B)$ is defined as follows.
The set $V(G_X)$ of nodes of $G_X$ 
is equal to the ground set $[m]$ of the matroids, and the set $E(G_X)$ of arcs 
is given by: $(i,j) \in E(G_X)$ if and only if one of the following holds:
\begin{itemize}
	\item $i \in X$, $j \not \in X$, and $i,j$ belong to a circuit of ${\bf M}(A)$.
	\item $i  \not \in X$, $j \in X$, and $i,j$ belong to a circuit of ${\bf M}(B)$.
\end{itemize}
Let $S_X = S_X(A)$ denote the subset of nodes $i\in E \setminus X$ such that $X \cup \{i\}$ is independent in ${\bf M}(A)$, and  $T_X = T_X(B)$ denote the subset of nodes $i \in E \setminus X$ such that $X \cup \{i\}$ is independent in ${\bf M}(B)$. See Figure \ref{fig:G_X} for $G_X$, $S_X$, and $T_X$.

\begin{figure}[t]
	\centering
	\includegraphics[width=8cm]{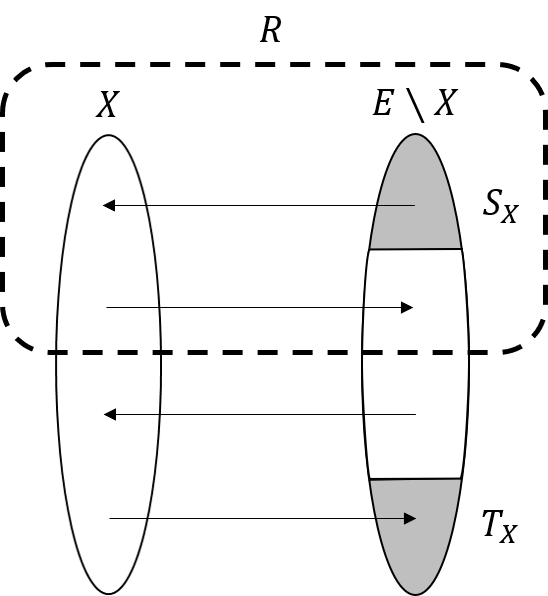}
	\caption{The auxiliary graph $G_{X}$}
	\label{fig:G_X}
\end{figure}

\begin{Lem}[{\cite{Edmonds79}}]\label{lem:optimality_matroidintersection}
	Let $X$ be a common independent set, and 
	let $R$ be the set of nodes reachable from $S_X$ in $G_X$.
	\begin{itemize}
		\item[{\rm (1)}] Suppose that $R \cap T_X \neq \emptyset$.
		For a shortest path $P$ from $S_X$ to $T_X$, 
		the set $X \triangle V(P)$ is a common independent set with 
		$|X \triangle V(P)| = |X| +1$.
		\item[{\rm (2)}] Suppose that  $R \cap T_X = \emptyset$. 
		Then $X$ is a maximum common independent set and $R$ attains $\min_{J \subseteq [m]} \rho_A(J) + \rho_B ([m] \setminus J)$.
	\end{itemize}
\end{Lem}

Here $\triangle$ denotes the symmetric difference.
According to this lemma, 
Edmonds' algorithm is as follows: 
\begin{itemize}
	\item Find a shortest path $P$ in $G_X$ from $S_X$ to $T_X$ (by BFS).
	\item If it exists, then replace $X$ by $X \triangle V(P)$, and repeat.
Otherwise, $X$ is a common independent set of the maximum cardinality.
\end{itemize}
 
In our case, 
the auxiliary graph $G_X$ and optimal matrices $K,L$ in (2$'$) are naturally 
obtained by applying elementary row operation to matrices $A,B$ as follows. 
Since $X$ is a common independent set,
both $A[X]$ and $B[X]$ have column full rank $|X|$.
Therefore, by multiplying nonsingular matrices $K$ and $L$ to $A$ and $B$ from left, respectively,  
we can make $A$ and $B$ diagonal in the position $X$, that is, 
for some injective maps $\sigma _A, \sigma _B : X \to [n]$, it holds $(KA)_{\sigma _A (i) i} = (LB)_{\sigma _B (i) i} = 1$ for $i \in X$ and other elements are zero.
Incorporating permutation matrices in $K,L$, we can assume $\sigma _A = \sigma _B = \sigma$.
Such matrices $KA$ and $LB$ are said to be {\em $X$-diagonal}.
Notice that these operations do not change the matroids ${\bf M}(A)$ and ${\bf M}(B)$. 
See Figure \ref{fig:AB}, where the columns and rows are permuted appropriately.

\begin{figure}[t]
			\centering
			\includegraphics[width=12cm]{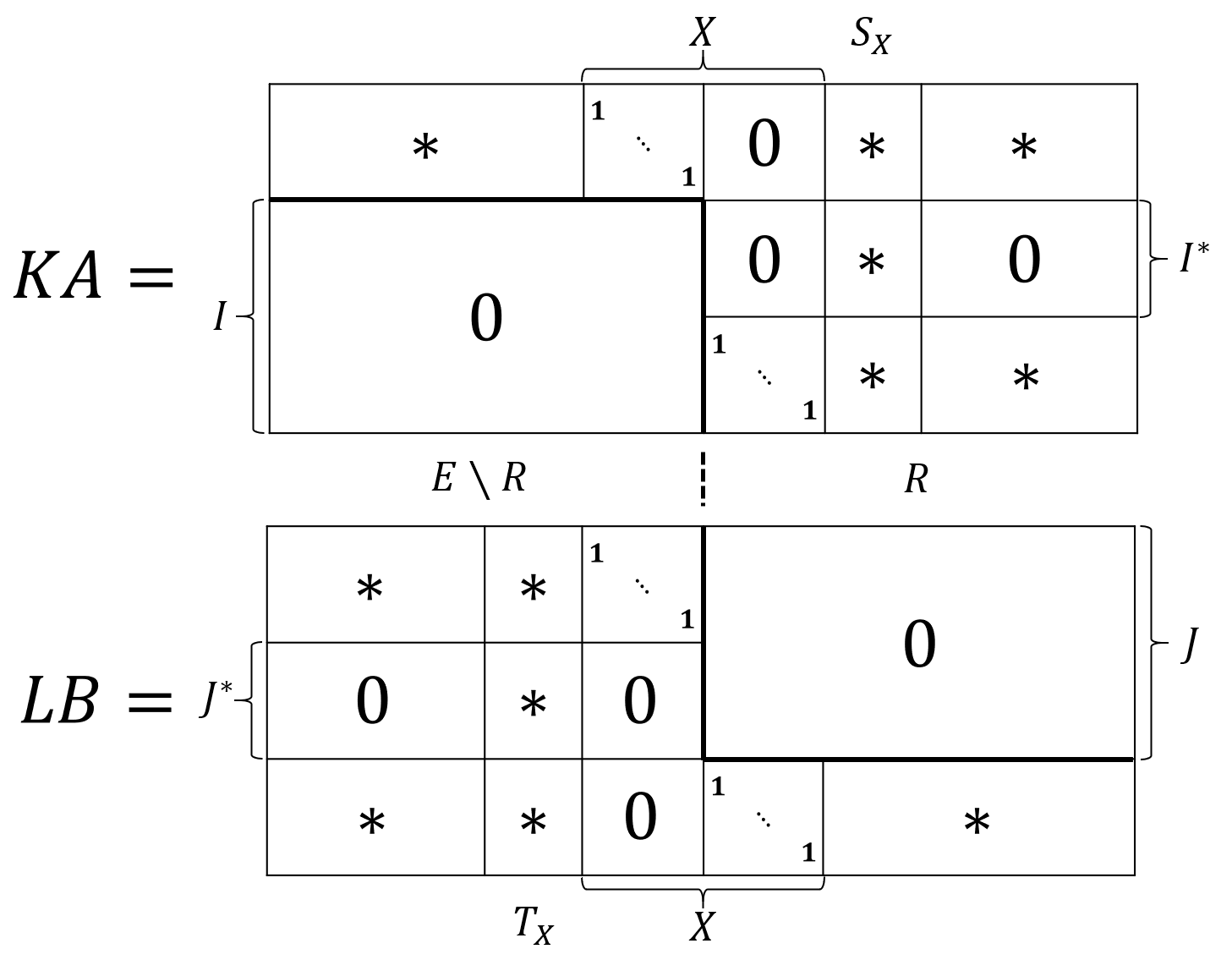}
			\caption{Matrices $A,B$ after elimination}
			\label{fig:AB}

		\begin{minipage}{0.5\hsize}
		\end{minipage}

			\centering
			\includegraphics[width=10cm]{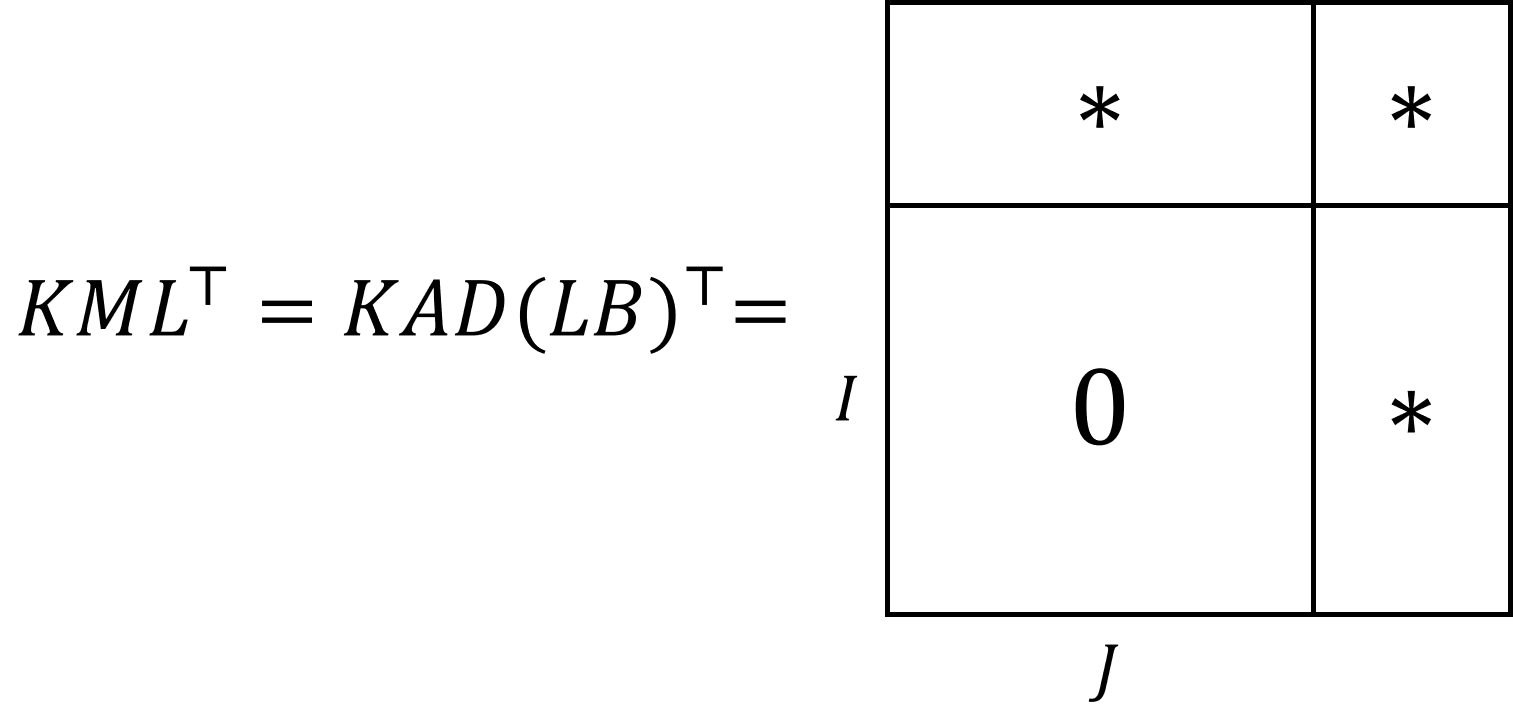}
			\caption{$KML^{\top}$ has zero submatrix $KML^{\top}[I,J]$, where $D$ is the diagonal matrix with diagonals $x_{1},x_{2},\dots,x_{n}$.}
			\label{fig:zeroblock}
\end{figure}

Then the auxiliary graph $G_X$ is constructed from the nonzero patterns of 
$KA$ and $LB$ as follows.
$S_X$ (resp. $T_X$) consists of nodes $i$
with $(KA)_{ki} \neq 0$ (resp. $(LB)_{ki} \neq 0$) for some $k \in [n] \setminus \sigma(X)$,
where $\sigma(X) = \{ j \in [n] \mid \exists i \in X, \sigma(i) = j \}$.
Additionally, for $i \in X$, arc $(i,j)$ (resp. $(j,i)$) exists if and only if 
$j \notin S_X$ and $(KA)_{\sigma(i)j} \neq 0$ (resp. $j \notin T_X$ and $(LB)_{\sigma(i)j}\neq 0$).

Moreover, 
in the case where $R \cap T_X = \emptyset$, 
the matrices $K,L^{\top}$ attain the maximum in (2$'$).
Indeed, define $I^{*}$, $J^{*}$, $I$ and $J$ by
\begin{eqnarray}
\label{eqn:IJ1}
I^{*} := [n] \setminus \sigma(X),\\
\label{eqn:IJ2}
J^{*} := [n] \setminus \sigma(X),\\
\label{eqn:IJ3}
I := \sigma(R \cap X) \cup I^{*},\\
\label{eqn:IJ4}
J := \sigma(X \setminus R) \cup J^{*}.
\end{eqnarray} 
Then the submatrix $(KML^{\top})[I,J]$ is an $(n - |X \setminus R|) \times (n - |R \cap X|)$ 
zero submatrix, where $|X| = 2n - (n - |X \setminus R| + n - |R \cap X|)$.
See Figure \ref{fig:zeroblock}.

\section{Algorithm}\label{sec:algorithm}

In this section, we consider the weighted linear matroid intersection problem.
In Section~\ref{subsec:formulation_weighted}, we formulate the problem as the computation of 
the degree of the determinant of 
a rational matrix associated with given two linear matroids and a weight. 
In Section~\ref{subsec:description}, we specialize {\bf Deg-Det} to present our algorithm for 
the weighted linear matroid intersection problem.
Its time complexity is analyzed in Section~\ref{subsec:analysis}, 
and its relation to Frank's algorithm is discussed in Section~\ref{subsec:Frank}.

\subsection{Algebraic formulation of weighted linear matroid intersection}\label{subsec:formulation_weighted}
Let $A,B$ be $n \times m$ matrices over $\QQ$ as in Section~\ref{subsec:formulation}, 
and let ${\bf M}(A)$ and ${\bf M}(B)$ be the associated linear matroids on $[m]$. 
We assume that both $A$ and $B$ have no zero columns.
In addition to $A,B$, 
we are further given integer weights 
$c_i \in \ZZ$ for $i \in [m]$.
The goal of the weighted linear matroid intersection problem is 
to maximize the weight $c(X) := \sum_{i \in X} c_i$ 
over all common independent sets $X$. 

Here we consider a restricted situation 
when the maximum is taken over all common independent sets of cardinality $n$.  
In this case, the maximum weight is interpreted as the degree of the determinant of the following $n \times n$ rational matrix $M$ defined by
\begin{equation*}
M := \sum_{i=1}^m  a_ib_i^{\top} x_i t^{c_i}.
\end{equation*}
\begin{Lem}
Suppose that $A$ and $B$ have row full rank.
The $\deg_t \det M$ is equal to the maximum of the weight $c(X)$ over all common independent sets $X$ of cardinality $n$. 
\end{Lem}
\begin{proof}
	As in the proof of Theorem~\ref{thm:MatroidIntersection}, by Binet-Cauchy formula applied to 
	$M$, we obtain 
	$\det M = \sum_{X \subseteq [m]: |X|=n}  \det A[X] \det B[X] t^{c(X)} \prod_{i \in X} x_i$, and
	\[
	\deg_t \det M = \max \{ c(X) \mid X \subseteq [m]: \det A[X] \det B[X] \neq 0\}.
	\]
\end{proof}	

\begin{Lem}[\cite{HH_degdet}]
	For the setting $M_i := a_ib_i^{\top} t^{c_i}$ $(i \in [m])$, 
	the algorithm {\bf Deg-Det} outputs $\deg_t \det M$.
\end{Lem}
\begin{proof}
	Consider step 2 of 
	{\bf Deg-Det}. Here $(PM_iQ)^0$ is written as $a_i^0{b_i^0}^{\top}$ for some $a_i^0,b_i^0 \in \QQ^n$; see (\ref{eqn:a_i^0}) and (\ref{eqn:b_i^0}) in the next subsection.
	In particular, $(PMQ)^0 = \sum_{i=1}^{m} a_{i}^{0}{b_{i}^{0}}^{\top} x_i$.
	Therefore, by Theorem~\ref{thm:MatroidIntersection}, $(PMQ)^0$ is nonsingular if and only if  the optimal value $r+s$ of (${\rm P}^0$) is at most $n$.
	Thus, if the algorithm terminates, then $(PMQ)^0$ is nonsingular 
	and $D^* = \deg_t \det M$ by Lemma~\ref{lem:optimality}.    
\end{proof}

\subsection{Algorithm description}\label{subsec:description}

Here we present our algorithm by specializing {\bf Deg-Det}.
The basic idea is to apply Edmonds' algorithm to solve the problem (P$^0$)
for $(PMQ)^0 = \sum_{i=1}^{m} (PM_iQ)^0 x_i$, where $PMQ$ is proper.
We first consider the case where 
$P$ and $Q$ are diagonal matrices represented as $P =(t^{\alpha})$ and $Q = (t^\beta)$ for some $\alpha,\beta \in \ZZ^n$.
In this case, $(PMQ)^0$ is explicitly written as follows.  
Observe that the properness of $PMQ$ is equivalent to 
\begin{equation}\label{eqn:properness}
\alpha_{k} + \beta_{\ell} + c_i \leq 0 \quad (i \in [m],k,\ell \in[n]: (a_i)_k(b_i)_{\ell} \neq 0).
\end{equation}
For $i\in [m]$, define $a_i^{0},b_i^{0} \in \QQ^n$ by
\begin{eqnarray}
\label{eqn:a_i^0}
&& (a_i^0)_k := \left\{
\begin{array}{ll}
(a_i)_k & {\rm if}\ \exists \ell \in [n], (a_i)_k(b_i)_{\ell} \neq 0, \alpha_k + \beta_{\ell} + c_i = 0, \\
0 & {\rm otherwise},
\end{array}
\right. \\
\label{eqn:b_i^0}
&& (b_i^0)_\ell := \left\{
\begin{array}{ll}
(b_i)_\ell & {\rm if}\ \exists k \in [n], (a_i)_k(b_i)_{\ell} \neq 0, \alpha_k + \beta_{\ell} + c_i = 0, \\
0 & {\rm otherwise}.
\end{array}
\right.
\end{eqnarray}
Then $(PM_iQ)^0 = a_i^0{b_i^0}^{\top}$. Namely we have 
\begin{equation*}
(PMQ)^0 = \sum_{i=1}^m a_i^0 {b_i^0}^\top x_i.
\end{equation*}
Therefore the step 1 of {\bf Deg-Det} can be executed by solving 
the unweighted linear matroid intersection problem 
for two matroids ${\bf M}(A^0)$ and ${\bf M}(B^0)$, 
where the matrices $A^0, B^0$ are defined by
\begin{equation*}
A^0 := (a_1^0\ a_2^0\ \cdots \ a_m^0),\ B^0 := (b_1^0\ b_2^0\ \cdots \ b_m^0).
\end{equation*}
The matrices $A^0,B^0$ have the following structure.
\begin{Lem}
\label{lem:weightA_0}
If $(a_i^0)_k \neq 0$ and $\alpha_{k'} = \alpha_k$,
then $(a_i^0)_{k'} = (a_i)_{k'}$.
If $(a_i^0)_k \neq 0$ and $\alpha_{k'} > \alpha_k$,
then $(a_i^0)_{k'} = (a_i)_{k'} = 0$.
The same properties holds for $B^0$ with $\beta$.
\end{Lem}
\begin{proof}
The former claim is immediate from the definition (\ref{eqn:a_i^0}).
For the latter claim, suppose to the contrary that $(a_i^0)_k$ and $(a_i^0)_{k'}$ are nonzero and $\alpha_{k'} > \alpha_{k}$.
Then for some $\ell, \ell'$, $(b_i^0)_\ell$ and $(b_i^0)_{\ell'}$ are nonzero with
$\alpha_k + \beta_{\ell} + c_i = \alpha_{k'} + \beta_{\ell'} + c_i = 0$
by the definition (\ref{eqn:a_i^0}).
Then, $(a_i)_{k'}(b_i)_{\ell} \neq 0$ and $\alpha_{k'} + \beta_{\ell} + c_i > \alpha_{k} + \beta_{\ell} + c_i = 0$.
This contradicts (\ref{eqn:properness}).
\end{proof}

Suppose that we are given a common independent set $X$ of ${\bf M}(A^0)$ and ${\bf M}(B^0)$.
According to Edmonds' algorithm (given after Lemma~\ref{lem:optimality_matroidintersection}), 
construct the residual graph $G_X^0 := G_X(A^0,B^0)$ 
with node sets $S_X^0 := S_X(A^0)$ and $T_X^0 := T_X(B^0)$.
Then we can increase $X$ or 
obtain $K,L$ that are optimal to the problem (P$^0$) (as was explained in the end of Section~\ref{subsec:formulation}). 

A key observation here is that $K$ and $L$ are
commuted with $(t^{\alpha})$ and $(t^{\beta})$, respectively:
\begin{equation}\label{eqn:commute}
K (t^{\alpha}) = (t^{\alpha}) K,\ L (t^{\beta}) = (t^{\beta}) L.
\end{equation}
Indeed, by Lemma~\ref{lem:weightA_0}, if $(a_i^0)_k$ and $(a_i^0)_{k'}$ are nonzero, then $\alpha_k = \alpha_{k'}$ holds.
Therefore, each elementary row operation for $A^0$ is done
between rows $k,k'$ with $\alpha_k = \alpha_{k'}$.
Consequently, the elimination matrix $K$ is a block diagonal matrix
in which the rows (columns) $k,k'$ in the same block have the same
$\alpha_k= \alpha_{k'}$. Then we can see the commutation
(\ref{eqn:commute}) as
\begin{eqnarray}
\label{eqn:Kt^alpha}
K(t^{\alpha}) &=& \left(
\begin{array}{cccc}
K_1 &      &  &   \\
    & K_2  &  &    \\
    &      & \ddots &    \\
    &      &  & K_k
\end{array}\right)
 \left(
 \begin{array}{cccc}
  t^{\bar \alpha_1}I &      &  &   \\
  & t^{\bar \alpha_2} I  &  &    \\
  &      & \ddots &    \\
  &      &  & t^{\bar \alpha_k} I
 \end{array}\right) \nonumber \\
& = & \left(
\begin{array}{cccc}
t^{\bar \alpha_1}I &      &  &   \\
& t^{\bar \alpha_2} I  &  &    \\
&      & \ddots &    \\
&      &  & t^{\bar \alpha_k} I
\end{array}\right)
\left(
\begin{array}{cccc}
K_1 &      &  &   \\
& K_2  &  &    \\
&      & \ddots &    \\
&      &  & K_k
\end{array}\right) = (t^{\alpha}) K,
\end{eqnarray}
where $\bar \alpha_1, \bar \alpha_2,\ldots, \bar \alpha_k$
are distinct values of $\alpha_1, \alpha_2,\ldots, \alpha_n$.
%

%
%
Therefore the update in step 3 of {\bf Deg-Det}
is done as $P \leftarrow (t^{\alpha + \kappa{\bf 1}_{I}})K$,
$Q \leftarrow L(t^{\beta -\kappa {\bf 1}_{[n]\setminus  {J}}})$.
Instead of doing such update,
we update $A,B$ as $A \leftarrow KA$,
$B \leftarrow L^{\top}B$,
which keeps $\deg_t \det M$,
and update $\alpha,\beta$
as $\alpha \leftarrow \alpha + \kappa{\bf 1}_{I}$,
$\beta \leftarrow \beta - \kappa {\bf 1}_{[n]\setminus  {J}}$.
Then $P,Q$ are always of the form $(t^\alpha), (t^{\beta})$,
and can be treated as exponent vectors $\alpha,\beta$,
where $- \deg_t \det P - \deg_t \det Q = - \sum_{i=1}^n(\alpha_i+ \beta_i)$.
Now the algorithm is written, without explicit references
to $P,Q,K,L$, as follows.
\begin{description}
	\item [Algorithm: Deg-Det-WMI]
	\item [Input:] $n \times m$ matrices $A = (a_1\ a_2\ \cdots \ a_m)$, $B = (b_1\ b_2\ \cdots \ b_m)$, 
	and weights $c_i \in \ZZ$ $(i=1,2,\ldots,m)$.
	\item [Output:] $\deg_t \det M$ for $M := \sum_{i=1}^m a_i b_i^{\top} x_i t^{c_i}$.
	\item [0:] $X = \emptyset$, 
	$\alpha := - \max_{i}{c_{i}}{\bf 1}$ and $\beta := {\bf 0}$.
	\item [1:] If $|X| = n$, then output $- \sum_{i=1}^n (\alpha_i + \beta_i)$ and stop. 
	Otherwise, according to (\ref{eqn:a_i^0}), (\ref{eqn:b_i^0}), decompose $A,B$ as $A = A^0 + A'$, $B = B^0 + B'$. 
	Apply elementary row operations to $A,B$ so that $A^0,B^0$ are $X$-diagonal forms.
	\item [2:] From $A^0,B^0$, 
	construct the residual graph $G^0_X$ and node sets $S^0_X,T^0_X$.
	Let $R^{0}$ be the set of nodes reachable from $S^0_X$ in $G^0_X$. 
	\begin{description}
		\item[2-1. {\rm If $R^{0} \cap T^0_X \neq \emptyset$:}] 	
		
		Taking a shortest path $P$ from $S^0_X$ to $T^0_X$, let $X \leftarrow X \triangle V(P)$, and go to step 1.
		\item[2-2. {\rm If $R^{0} \cap T^0_X = \emptyset$:}] 
		Then $R^0$ determines the zero submatrix $((t^{\alpha}) M (t^{\beta}))^0[I,J]$ of maximum size $|I| + |J| (> n)$ by (\ref{eqn:IJ3}) and (\ref{eqn:IJ4}); see also Figures \ref{fig:AB} and \ref{fig:zeroblock}.
		Letting $\alpha \leftarrow \alpha + \kappa {\bf 1}_{I}$, $\beta \leftarrow \beta - \kappa {\bf 1}_{[n] \setminus J}$,
		increase $\kappa$ from $0$ until 
		a nonzero entry appears in the zero submatrix.
		If $\kappa = \infty$ or $-\sum_{i=1}^{n}(\alpha_{i}+\beta_{i}) < n\min_{i}{c_{i}}$, then output $-\infty$ and stop.
		Otherwise go to step 1. 
	\end{description} 
\end{description}
The step 2 in this algorithm is
essentially Edmonds' algorithm
to solve the unweighted matroid intersection problem
for two matroids ${\bf M}(A^0)$, ${\bf M}(B^0)$ and an initial
common independent set $X$. It turns out below that $X$ is
actually commonly independent for ${\bf M}(A^0)$ and ${\bf M}(B^0)$.
Assuming this,
it is clear that, in step 2-1, $X$ increases and is a common
independent set in the next step 1,
and that, in step 2-2, $X$ is a maximum common independent set
and a maximum-size zero submatrix of $((t^{\alpha})M(t^{\beta}))^0 =
\sum_{i=1}^m a_i^0 {b_i^0}^{\top} x_i$ is obtained accordingly.
After the update of $\alpha,\beta$,
$A^0$ and $B^0$ are changed so that
$A^0[[n] \setminus I, R^0]$ and $B^0[[n] \setminus J, E \setminus R^0]$
become zero blocks, and $A^0[I, E \setminus R^0]$ or $B^0[J, R^0]$ has
nonzero entries; see Figure~\ref{fig:AB'} in Section~\ref{subsec:analysis}. Other
parts are unchanged.
In particular, both $A^0[X]$ and $B^0[X]$
are lower triangular matrices (by row/column permutations).
Therefore $X$ keeps commonly independent
for new matroids ${\bf M}(A^0)$ and ${\bf M}(B^0)$ in the next step 1.
If $|X| = n$, then this is in the situation where $(PMQ)^0$ is
nonsingular,  and hence
the algorithm correctly outputs $\deg_t \det M$ as $- \deg_t \det P - \deg_t
\det Q = - \sum_{i=1}^n (\alpha_i+\beta_i)$.
If the singularity of $M$ is detected, e.g., $\deg_t \det M < n \min_{i
\in [m]} c_i$,
then it outputs $-\infty$.

Moreover, $X$ is always a common independent set of ${\bf M}(A)$
and ${\bf M}(B)$ having the maximum weight
among all common independent sets of cardinality $|X|$.
Therefore {\bf Deg-Det-WMI} can obtain a maximum weight
independent set (of arbitrary cardinality) by adding the following
procedure.
\begin{itemize}
\item After the update of $X$ in step 2-1, for $k=|X|$,
output $X_k :=X$ as a maximum weight common independent set of cardinality $k$
for ${\bf M}(A)$ and ${\bf M}(B)$.
\item After the termination of the algorithm,
output $X^*$ from $X_0,X_1,\ldots,X_n$ having the maximum weight
$c(X_i)$, where $X_0 := \emptyset$ and $c(X_k) := -\infty$ if $X_k$ is
undefined. Then $X^*$ is a maximum weight common independent set for ${\bf M}(A)$ and ${\bf M}(B)$.
\end{itemize}
We show this fact by using the idea of {\em weight splitting}~\cite{Frank81}.
%
%
\begin{Lem}\label{lem:independent}
	In step 1, define weight splitting $c_i = c_i^1 + c_i^2$ for each $i \in [m]$ by
	\begin{eqnarray}\label{eqn:c_i^1}
	c_i^1 & := & c_i - c_i^2,\\
	\label{eqn:c_i^2}
	c_i^2 & := & - \max \{ \beta_\ell \mid \ell \in [n]: (b_i)_\ell \neq 0 \}.
	\end{eqnarray}
	Then $X$ is a common independent set of ${\bf M}(A)$ 
	and ${\bf M}(B)$ such that $c^{1}(X) = \max \{ c^1(Y) \mid Y \in {\cal I}(A), |Y|  = |X|\}$
	and  $c^{2}(X) = \max \{ c^2(Y) \mid Y \in {\cal I}(B), |Y| = |X|\}$. Thus $X$  
	maximizes the weight $c(X)$ over all common independent sets of cardinality $|X|$.
\end{Lem}
\begin{proof}
We first verify that $X$ is a common independent set of ${\bf M}(A)$ and ${\bf M}(B)$.
We may assume $X = \{1,2,\ldots,h\}$.
Since $X$ is commonly independent of ${\bf M}(A^0)$ and ${\bf M}(B^0)$,
we can assume that $A^0[[h],X] = B^0[[h],X] = I$ in the $X$-diagonal forms.
Then $I^* = J^* = \{h+1,\ldots,n\}$; recall (\ref{eqn:IJ1}) and (\ref{eqn:IJ2}).
We can further assume that $\alpha_1 \geq \alpha_2 \geq \cdots \geq \alpha_h$ and $\beta_1 \geq \beta_2 \geq \cdots \geq \beta_h$.
By Lemma~\ref{lem:weightA_0}, $A[X]$ and $B[X]$ are lower-triangular matrices with nonzero diagonals.
Hence $X$ is commonly independent for ${\bf M}(A)$ and ${\bf M}(B)$.

Next 
we make some observations to prove the statement. 	
Observe from the definition (\ref{eqn:a_i^0}) (\ref{eqn:b_i^0}) (\ref{eqn:c_i^1}) (\ref{eqn:c_i^2}) and the properness (\ref{eqn:properness}) that
\begin{eqnarray}
c^1_i &\leq & - \alpha_k \quad ( \forall k: (a_i)_k \neq 0), \label{eqn:c^1_i} \\ 
c^2_i &\leq & - \beta_\ell \quad (\forall \ell: (b_i)_\ell \neq 0), \label{eqn:c^2_i}
\end{eqnarray}
and 
\begin{equation}\label{eqn:c^1_ic^2_i=}
c_i^1  = - \alpha_k, \quad c_i^2 = - \beta_{\ell} \quad (\forall k,\ell: (a_i^0)_k (b_i^0)_{\ell} \neq 0).
\end{equation}
We also observe
\begin{equation}\label{eqn:I*J*}
\max_{k \in [n]} \alpha_k = \alpha_{k'} \, (\forall k' \in I^*), \quad  \max_{\ell \in [n]} \beta_\ell = \beta_{\ell'} \, (\forall \ell' \in J^*).
\end{equation}
This follows from  
the way of update $\alpha \leftarrow \alpha + {\bf 1}_{I}$, 
$\beta \leftarrow \beta - {\bf 1}_{[n] \setminus J}$ with the initialization $\alpha = -\max_{i} c_{i} {\bf 1}$, 
$\beta = {\bf 0}$ of the algorithm, and the fact
that both $I^* \subseteq I$ and $J^* \subseteq J$ 
monotonically decrease.

Finally we prove that $X$ maximizes both weights $c^1$ and $c^2$ for ${\bf M}(A)$ and ${\bf M}(B)$, respectively. 
It suffices to show
\begin{eqnarray}
c^1(X) &\geq & c^1(X \cup \{i\} \setminus \{j\}) \quad (i \not \in X,j \in X: X \cup \{i\} \setminus \{j\} \in {\cal I}(A)), \label{eqn:opt_matroid1} \\
c^2(X) &\geq & c^2(X \cup \{i\} \setminus \{j\}) \quad (i \not \in X,j \in X: X \cup \{i\} \setminus \{j\} \in {\cal I}(B)). \label{eqn:opt_matroid2}
\end{eqnarray}
Indeed, this is the well-known optimality criterion of 
the maximum weight independent set problem on a matroid.
Take $i,j$ with $X \cup \{i\} \setminus \{j\} \in {\cal I}(A)$.  
If there is a nonzero element $(a_{i})_{k^*}\neq 0$ for some $k^* \in I^*$, 
then by (\ref{eqn:c^1_i}) and (\ref{eqn:I*J*})
it holds $c_i^1 \leq -\alpha_{k^*}  \leq -\alpha_j = c_j^1$, 
where the equality follows from (\ref{eqn:c^1_ic^2_i=}) and $(a_j^0)_j = 1$,
and thus (\ref{eqn:opt_matroid1}) holds. 
Suppose not.
Let $k \in [h]$ be the smallest index such that $(a_i)_k \neq 0$. 
Then $c^1_i \leq - \alpha_k$.  
Now $A[[h],X]$ is lower triangular.
Additionally, by Lemma~\ref{lem:weightA_0} and (\ref{eqn:I*J*}), $A[I^*,X] = A^0[I^*,X]$ 
is a zero matrix. 
Therefore,  
it must hold $j \geq k$ 
for $i,j$ to belong to a circuit in $X \cup \{i\}$. 
Hence, $c^1_j = - \alpha_j \geq - \alpha_k \geq c^1_i$. 
Thus (\ref{eqn:opt_matroid1}) holds.
(\ref{eqn:opt_matroid2}) is similarly shown.
\end{proof}

\subsection{Analysis}\label{subsec:analysis}
We analyze the time complexity of {\bf Deg-Det-WMI}.
It is obvious that if $R^{0} \cap T^0_X \neq \emptyset$ (step 2-1) occurs, 
then $X$ increases and hence the rank of $((t^{\alpha})M(t^{\beta}))^0$ increases.
Therefore the algorithm goes to step 2-1 at most $n$ times. 
The main analysis concerns step 2-2, particularly, how   
nonzero entries appear, 
how they affect $A^0$, $B^0$, and $G _X ^0$,  
and how many times these scenarios 
occur until $R^0 \cap T_X^0 \neq \emptyset$.

As $\kappa$ becomes positive, 
the submatrix $((t^{\alpha})M(t^{\beta}))^0[[n] \setminus I, [n] \setminus  J]$ becomes a zero block, since the degree of each element of $(t^{\alpha})M(t^{\beta})[[n] \setminus I, [n] \setminus  J]$ decreases. 
Accordingly,   
$A^0[[n] \setminus I, R^0]$ and
$B^0[[n] \setminus  J, E \setminus R^0]$ become zero blocks;
see Figure \ref{fig:AB'}.
\begin{figure}[t]
	\centering
	\includegraphics[width=12cm]{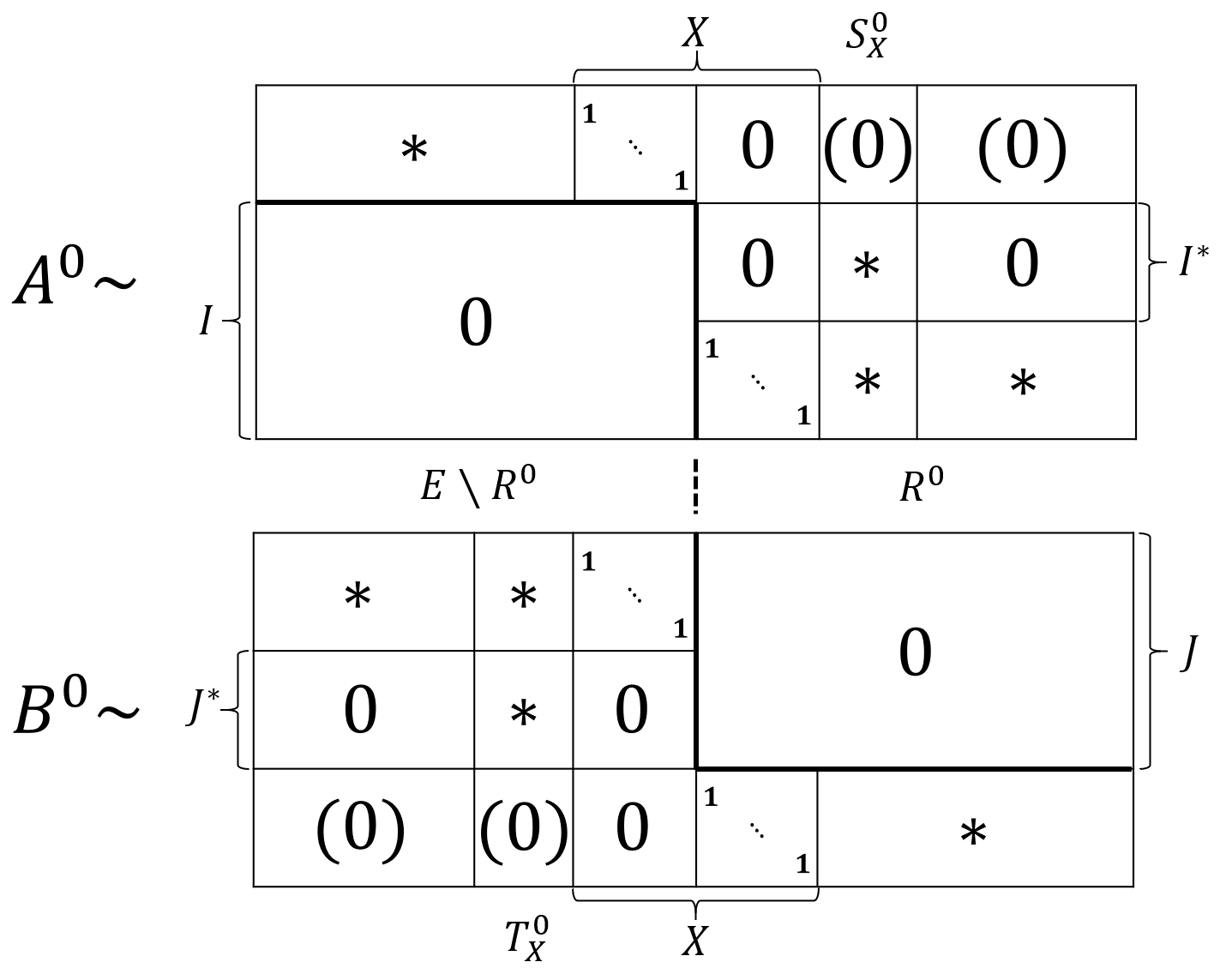}
	\caption{Change of $A^0,B^0$}
	\label{fig:AB'}
\end{figure}
Then, in $G^0_{X}$, all arcs entering $R^0$ disappear. 
Namely increasing $\kappa$ only removes arcs entering to $R^0$ and does not change the other parts.
%

Next we analyze the moment when a non-zero element appears in $((t^{\alpha})M(t^{\beta}))^0 [I,J]$.
Then, in the next step~1,
it holds
\[
(a_i^0)_k (b_i^0)_\ell \neq 0 
\]
for some $i \in [m]$, $k \in I$, $\ell \in J$.
In this case, a new nonzero element appears 
in the $i$-th column of $A^0$ or $B^0$.
\begin{description}
	\item[{\rm (a-1)}] If $i \not \in R^0$ and $i \in X$:
		In the next step 1,
		Gaussian elimination for $A^0$ (and $A$) makes 
		the new nonzero element $(a_i^0)_k = (a_i)_k$ zero. 
		Since $A^0[[n] \setminus I, R^0] = O$, this does not affect 
		$A^0[R^0]$. Therefore $R^0$ is still reachable from $S_X^0$.
		There may appear nonzero elements in $A^0[I, E \setminus R^0]$, 
		which will make $R^0$ or $S_X^0$ larger in the next step~2. 
		\item[{\rm (a-2)}] If $i \not \in R^0$ and $i \not \in X$:
		By $(a^0_i)_k  \neq 0$, 
		if $k \in I^*$, then $i$ is included to $S^0_X$.
		Otherwise there appears an arc in $G_X^0$ from $X \cap R^0$ to $i$.
		For the both cases, $i$ is included to $R^0$.
	    By $\ell \in J$, if $\ell \in J^*$, then $i$ belongs to $T_X^0$.
	    Otherwise there is an arc from $i$ to $X \setminus R^0$.  
        Thus, $R^0 \cap T_X^0$ becomes nonempty if $\ell \in J^*$,
	    and $|X \cap R^0|$ increases if $\ell \in J \setminus J^*$.
	\item[{\rm (b-1)}] If $i \in R^0$ and $i \in X$:
		Similar to the analysis of (a-1) above, 
		Gaussian elimination for $B^0$ makes $(b^0_i)_k = (b_i)_k$ zero, and  
		$R^0$ and $T_X^0$ increase or do not change.
		 
		\item[{\rm (b-2)}] If $i \in R^0$  and $i \not \in X$:
		By $(b^0_i)_\ell  \neq 0$, 
		if $\ell \in J^*$, then $i$ is included to $T_X^0$, and $R^0 \cap T_X^0 \neq \emptyset$.
		Otherwise there appears an arc from $i$ to $X \setminus R^0$,
		and $|X \cap R^0|$ increases.
\end{description}

Therefore, 
if the case (a-2) or (b-2) occurs, then 
$T_X^0 \cap R^0 \neq \emptyset$ or $|X\cap R^0|$ increases.
After $O(n)$ occurrences of the cases (a-2) and (b-2),  $T_X^0 \cap R^0$ becomes nonempty and  
$|X|$ increases.
When $X$ is updated, 
Gaussian elimination constructs the $X$-diagonal forms of $A^0,B^0$
in $O(mn^2)$ time.

We analyze the occurrences of (a-1) and (b-1).
When $(a_i^0)_k$ becomes nonzero for some $i \in X \setminus R^0, k \in I$,  
it is eliminated by the row operation, and $(a_i^0)_k = (a_i)_k$ never becomes nonzero.
Therefore, (a-1) and (b-1) occur at most $O(n |X|)$ time until $X$ is updated, 
where the row operation is executed in $O(m)$ time per each occurrence.  
The total time for the elimination is $O(nm|X|)$.
The augmentation $\kappa$ and 
the identification of the next nonzero elements are computed 
in $O(nm)$ time by searching nonzero elements in $A,B$, which
is needed for each time one of (a-1), (a-2), (b-1), and (b-2) occurs.
Thus, by the naive implementation, {\bf Deg-Det-WMI} runs in $O(mn^4)$ time. 

We improve this complexity to $O(mn^3 \log n)$ as follows.
Observe first that $\kappa$ is given by
\[
\kappa = - \max \{ c_i + \alpha_{k} + \beta_{\ell} \mid i \in [m], k \in I, \ell \in J: (a_i)_k(b_i)_{\ell} \neq 0 \}.
\]
The main idea is to sort 
indices $(i,k,\ell) \in [m] \times I \times J$ according to $c_i + \alpha_{k} + \beta_{\ell}$ and keep in a binary heap the potential indices that attain $\kappa$.
Notice that even if $(a_i)_k(b_i)_\ell$ is zero in a moment, 
it will become nonzero by row operations in (a-1) and (b-1) 
and can appear in $((t^\alpha)M(t^\beta))^0[I,J]$ later. 
On the other hand, any index $(i,k,\ell)$ with $c_i + \alpha_{k} + \beta_{\ell} > 0$
keeps $(a_i)_k(b_i)_\ell = 0$ and is irrelevant until $X$ is updated.

Suppose now that $X$, $A^0$, $B^0$, and $G_X^0$ were updated in step 1.
By BFS for $G_X^0$, we determine the reachable set $R^0$ and the index sets $I,J$.
%
We can sort $c_i + \beta_{\ell}$ $(i \in [m], \ell \in J)$
in $O(mn \log m)$ time, which is improved to $O(mn \log n)$ time as follows.
By sorting $c_i$ $(i \in [m])$ in $O(m \log m)$ time,
we obtain $|J|$ sorted lists of $c_i + \beta_{\ell}$ $(i \in [m])$ for
$\ell \in J$.
By keeping the head elements of these sorted lists in a heap,
the whole sorted list can be obtained in $O(nm \log |J|)$, as in the merge sort.

From the sorted list, we construct an array $p$ such that the $e$-th entry
$p[e]$ has all indices $(i,\ell)$
with $e$-th largest $c_i + \beta_{\ell}$ as a linked list. 
For each $k \in I$, let $p_k$ denote the copy of the array $p$, 
where $p_k[e]$ also has the value $v_{k,e} := c_i + \alpha_k + \beta_\ell$
for indices $(i,\ell)$ in $p_k[e]$. 
By the {\em head index} of $p_k$
(relative to $\alpha,\beta, I, J$), we mean the minimum index $e_k$ 
such that $p_k[e_k]$ has the value $v_{k,e_k}$ less than $0$ 
and an index $(i,\ell)$ with $\ell \in J$, where $J$ will decrease later. 
Notice that if $p_k[e]$ has the value $v_{k,e} \geq 0$, 
then $(a_i)_k (b_i)_\ell = 0$ for all indices $(i, \ell)$ in $p_k[e]$.
Construct a binary (max) heap consisting 
of the pointers to the head indices $e_k$ for all $k \in I$, 
where the key is the value $v_{k,e_k}$ of $p_k[e_k]$.
In the construction of the heap, 
if the key $v_{k,e_k}$ of 
a node is equal to the key $v_{k',e_{k'}}$ of its parent node,   
then the two nodes are combined as a single node and
the corresponding pointers are also combined as a single list.  
Then,  by referring to the root of the heap, 
we know all indices $(i,k,\ell) \in [m] \times I \times J$ 
having the maximum negative value. 
Increase $\kappa$ to the negative of this value
(i.e., $\alpha \leftarrow \alpha + \kappa {\bf 1}_{I}$,
$\beta \leftarrow \beta - \kappa {\bf 1}_{[n] \setminus J}$).
If the root has no index $(i,k,\ell)$ 
with $(a_i)_k (b_i)_\ell \neq 0$, then delete the root from the heap, 
update the head index of each $p_k$ indicated 
by the (deleted) root, 
and add the pointers of new head indices to the heap.
Suppose that the root has an index $(i,k,\ell)$ with $(a_i)_k (b_i)_\ell \neq 0$;  
then $\kappa = - c_i - \alpha_k - \beta_{\ell}$.
If $i \in X$ then
execute the row operation to make $(a_i)_k (b_i)_\ell$ zero.
As mentioned, once $(a_i)_k (b_i)_\ell$ becomes zero by the row operation,
it never becomes nonzero. Here $(a_{i'})_{k'} (b_{i'})_{\ell'}$
for another index $(i',k',\ell')$ in the root
may become nonzero from zero,
which is eliminated in the next if $i' \in X$.
Therefore, together with doing such row operations,
after looking the indices in the root at most twice,
the root has no index $(i,k,\ell)$
with $i \in X$ and $(a_i)_k (b_i)_\ell \neq 0$.
Suppose that 
there is $(i,k,\ell)$ with $i \not \in X$ and $(a_i)_k (b_i)_\ell \neq 0$.
Then $G_X^0$, $R^0$, $I$, and $J$ are updated. 
In particular, $I$ increases and $J$ decreases.  
For each newly added $k \in I$, construct array $p_{k}$ (from $p$), identify the head index of $p_k$, 
and add the pointer to the heap.
In this way, until $X$ increases, 
each index $(i,k,\ell)$ is referred to at most twice, and 
the heap is updated in $O(\log n)$ time per the reference. 
In total, $O(mn^2 \log n)$ time is required.
Thus we have:
\begin{Thm}
	Algorithm {\bf Deg-Det-WMI} runs in ${\rm O}(mn^3 \log n)$ time.
\end{Thm}

\subsection{Relation to Frank's algorithm}\label{subsec:Frank}

In this subsection, we reveal the relation between our algorithm {\bf Deg-Det-WMI} and Frank's weight splitting algorithm \cite{Frank81}.
We show that the common independent sets $X$ obtained by {\bf Deg-Det-WMI} are the same as the ones obtained by a slightly modified version of Frank's algorithm.
This means in a sense that {\bf Deg-Det-WMI} is a nonstandard specialization of Frank's algorithm to linear matroids. 

Let us briefly explain Frank's algorithm; 
our presentation basically follows \cite[section 13.7]{Korte12}.
His algorithm keeps a weight splitting $c_{i}=c_{i}^{1}+c_{i}^{2}$ for each $i\in E$ and a common independent set $X$ such that $X$ is maximum for both $c_{i}^{1}$ and $c_{i}^{2}$ over all common independent sets of cardinality $|X|$.
\begin{description}
\item[0:] $c _i ^1:=c_i$, $c _i ^2 :=0$ for $i \in E$ and $X:=\emptyset$.
\item[1:] Applying elementary row operations to $A$, $B$, construct the residual graph $G_{X}$, and node sets $S_{X}$, $T_{X}$ as in Section 2.2.
\item[2:] From the weight splitting $c= c^{1} + c^{2}$, construct subgraph ${\bar G}_{X}$ of $G_X$ and node subsets $\bar S _X \subseteq S _X$, $\bar T _X \subseteq T _X$ by:
${\bar G}_{X}$ consists of arcs $ij$ with $i \in X \not \ni j$ and $c_{i}^{1}=c_{j}^{1}$ or $i \notin X\ni j$ and $c_{i}^{2}=c_{j}^{2}$, and
\begin{eqnarray}
\label{eqn:SX}
{\bar S}_{X}:= \{ i \in S_{X} \mid \forall j\in S_{X},c_{i}^{1} \geq c_{j}^{1}\},\\ 
\label{eqn:TX}
{\bar T}_{X}:= \{ i \in T_{X} \mid \forall j\in T_{X},c_{i}^{2} \geq c_{j}^{2}\}.
\end{eqnarray}
\item[3:] Let $\bar R$ be the set of nodes reachable from ${\bar S}_{X}$ in ${\bar G}_{X}$.
\item[4-1:] If $\bar R \cap \bar T _{X} \neq \emptyset$, for a shortest path $P$ from $\bar S _{X}$ to $\bar T _{X}$, replace $X$ by $X \Delta V(P)$; go to step~1.
\item[4-2:] If $\bar R \cap \bar T _{X} = \emptyset$, then let $c_{i}^{1}:=c_{i}^{1}-\epsilon$, $c_{i}^{2}:=c_{i}^{2}+\epsilon$ for $i \in \bar R$, and increase $\epsilon$ from $0$ until $\bar R$ increases.
If $\epsilon = \infty$, then output $-\infty$ and stop.
Go to step~2.
\end{description}


We consider a modified update of the weight splitting.
Let $\bar R'$ be the subset of nodes $i \in E \setminus (X \cup \bar R )$ such that 
all arcs leaving $i$ enters $X \cap \bar R$.
Then the step~4-2 can be replaced by the following:
\begin{description}
\item[4-2$'$:] If $\bar R \cap \bar T _{X} = \emptyset$, then let $c_{i}^{1}:=c_{i}^{1}-\epsilon$, $c_{i}^{2}:=c_{i}^{2}+\epsilon$ for $i \in \bar R \cup \bar R'$, and increase $\epsilon$ from $0$ until $\bar R$ increases or $\bar R'$ changes.
If $\epsilon = \infty$, then output $-\infty$ and stop.
Repeat until $\bar R$ increases and go to step~2.
\end{description}
One can easily check that $X$ keeps the optimality (\ref{eqn:opt_matroid1}), (\ref{eqn:opt_matroid2}) in the modified update.
Hence, the modified algorithm using {\bf 4-2$'$} is also correct. 

We prove that $G_{X}^{0}$, $S _X ^0$, $T _X ^0$ in our algorithm and $\bar G _{X}$, $\bar S _X$, $\bar T _X$ in modified Frank's algorithm are the same up to an obvious redundancy.
Here an arc in $\bar G _{X}$ is said to be {\it redundant} if it leaves a node $i$ that has no arc entering $i$.
\begin{Prop}
Suppose that $X$, $\alpha$ and $\beta$ are obtained in an iteration of {\bf Deg-Det-WMI}. Define weight splitting $c_{i} = c_{i}^{1}+c_{i}^{2}$ by (\ref{eqn:c_i^1}), (\ref{eqn:c_i^2}) and $\bar G _{X}$, $\bar S _{X}$ and $\bar T _{X}$ by (\ref{eqn:SX}), (\ref{eqn:TX}).
Then we have the following:
	\begin{itemize}		
		\item[{\rm (1)}] $G_{X}^{0}$ is equal to the subgraph of $\bar G _{X}$ obtained by removing redundant arcs.
		\item[{\rm (2)}] $S_{X}^{0}$ is equal to $\bar S _{X}$.
		\item[{\rm (3)}] $T_{X}^{0}$ is equal to the subset of $\bar T_{X}$ obtained by removing isolated nodes.
	\item[{\rm (4)}] $R^{0}$ is equal to $\bar R$.
	\item[{\rm (5)}] The total sum of increases $\kappa$ until $R^{0}$ changes is equal to that of increases $\epsilon$ until $\bar R$ changes in the modified Frank's algorithm.
	\end{itemize}
\end{Prop}

\begin{proof}
Recall (the proof of) Lemma \ref{lem:independent} that $X$ is a common independent 
set of ${\bf M}(A)$ and ${\bf M}(B)$. 
Suppose that $X = \{1,2,\ldots,h\}$ and
$A^0[[h], X] = B^0[[h],X] = I$ with $\alpha_1 \geq \alpha_2 \geq \cdots \geq \alpha_h$ and $\beta_1 \geq \beta_2 \geq \cdots \geq \beta_h$.
    Observe first that $S^0_X \subseteq S_X$.
Indeed, from Lemma~\ref{lem:weightA_0} and (\ref{eqn:I*J*}), $A[I^*,X]$ is a zero matrix.
Therefore, if $a_{i}^{0}$ has a nonzero vector in a row in $I^*$, i.e., $i \in S_X^0$, then $a_{i}$ is independent from $a_{i'}$ $(i' \in X)$, i.e.,  $i \in S_X$. 
Consider the weight splitting of nodes in $S_X^0$.
For $i \in S_X^0$, $c _i ^1 = - \alpha_k\,(k \in I^*)$, and
$- \alpha_k \geq c_j^1$ for $j \in S_X$ by (\ref{eqn:I*J*}).
Thus $S_X^0 \subseteq \bar S_X$.
Also, for any $i' \in S_X \setminus S_X^0$,
$a^0 _{i'}$ is a zero vector.
Indeed, it holds $(a _{i'})_{k^*} \neq 0 = (a^0 _{i'})_{k^*}$ for some
$k^* \in I^*$.
This means $\alpha_{k^*} + \beta_{\ell} + c_{i'} < 0$
for all $\ell \in [n]$ with $(a _{i'})_{k^*}(b_{i'})_\ell \neq 0$.
By (\ref{eqn:I*J*}), it holds $\alpha_k \leq \alpha_{k^*}$
for all $k \in [n]$, and
$\alpha_{k} + \beta_{\ell} + c_{i'} < 0$ for all $k,\ell$ with $(a
_{i'})_{k}(b_{i'})_\ell \neq 0$, which implies $a^0 _{i'} = {\bf 0}$.
Thus, it holds $c_{i'}^1 < - \alpha_k$ for $k \in I^*$.
Then $c_{i'}^1 < - \alpha_k = c_i^1$ for $i \in S_X^0$. 
Thus we have (2).
  
Showing (3) is similar. 
As above, we see that $T_X^0 \subseteq T_X$ and for $i \in T _X ^0$, $c _i ^2 = -\beta_\ell$ ($\ell \in J^*$).
Then $T_X^0 \subseteq \bar T_X$.
Let $i \in T_X \setminus T_X^0$.
Then  $b _i ^0$ is a zero vector, and so is $a _i ^0$.
Suppose that arc $ji$ for $j \in X$ exists in $G_X$.
Recall that $A[[h],X]$ and $B[[h],X]$ are lower triangular.
Then $j$ is at least the minimum index $k$ with $(a_i)_k \neq 0$.
Then for $\ell \in J^{*}$, $c_j^1 = - \alpha_j \geq - \alpha_k > c_i + \beta_\ell 
= c_i - c_i^2 = c_i^1$, where the strictly inequality follows from the fact that $a _i ^0$ and $b _i ^0$ are zero vectors.
Then $ji$ does not exist in $\bar G_X$. 
Similarly, arc $ij$ does not exist in $\bar G_X$
and this means $i$ is an isolated node.
Thus we have (3).

    Next we compare $G_X^0$ and $\bar G_X$ to prove (1) and (4).
    Consider a node $i \in E \setminus X$
    such that $a_{i}^{0}$ and $b_i^0$ are nonzero.
    Suppose that arc $ki$ exists in $G_X^0$, i.e., $(a_i^0)_k \neq 0$
for $k \in [h] = X$.
    Then $c_i^1 = - \alpha_k = c_k^1$.
    We show that $ki$ exists also in $\bar G _X$.
    Since $\alpha_j \geq \alpha_k$ $(j \neq k)$ implies $(a_k)_j = 0$ by Lemma~\ref{lem:weightA_0},
    Gaussian elimination making $A$ $X$-diagonal
    does not affect $(a_{i})_k$.
    Thus the arcs $ki$ exists in $G_X$ and in $\bar G_X$.
    Similarly, if $i \ell$ exists in $G_X^0$, then $i \ell$ exists in
$\bar G_X$.
    Therefore, for any node $i \in E \setminus X$ with nonzero
$a_{i}^{0}$,$b_i^0$,
    the arcs incident to $i$ are the same in $G_X^0$ and $\bar G_X$.
    Consider a node $i \in E \setminus X$
    such that $a_{i}^{0}$ and $b_{i}^{0}$ are zero vectors.
    In $G_X^0$, there are no arcs incident to $i$.
    For $k \in X, \ell \in [n]$ with $(a_{i})_k(b_i)_{\ell} \neq 0$,
    it holds $c_k^1 = - \alpha_k > \beta_\ell + c_i \geq - c_i^2 + c_i = c_i^1$.
    This means that arcs $ki$ entering $i$ do not exist in $\bar G_X$,
and thus arcs $i \ell$ leaving $i$ are redundant.
Thus we have (1).
    From (1), (2), and (3), we have (4).
 
  Finally we prove (5).
   The step 2-2 in {\bf Deg-Det-WMI} changes $\alpha,\beta$ as
   $\alpha \leftarrow \alpha + \kappa {\bf 1}_{I}$, $\beta \leftarrow
\beta - \kappa {\bf 1}_{[n] \setminus J}$.
   We analyze the corresponding change of the weight splitting $c=c^1+
c^2$ defined by (\ref{eqn:c_i^1}),$\,$(\ref{eqn:c_i^2}).
   Consider $i \in [m]$ such that $a_{i}^{0}$ and $b_{i}^{0}$ are
nonzero vectors.
   Suppose that $i \in R^0 = \bar R$.
   Then $a_i^0$ and $b_i^0$ have nonzero entries
   in a row in $I$ and in $[n] \setminus J$, respectively;
   see Figure \ref{fig:AB'}.
   Therefore $c_i^1 = - \alpha_k$ for some $k \in I$
   and $c_i^2 = - \beta_\ell$ for some $\ell \in [n] \setminus J$,
   and $c_i^1,c_i^2$ are changed as
   $c_{i}^{1} \leftarrow c_{i}^{1} - \kappa$, $c_{i}^{2} \leftarrow
c_{i}^{2} + \kappa$.
   Suppose that $i \notin R^0$.
   Then $a_i^0$ and $b_i^0$ have nonzero entries
   in a row in $[n] \setminus I$ and in $J$, respectively.
   In particular, $c_i^1 = - \alpha_k$ for some
$k \in [n] \setminus I$
   and $c_i^2 = - \beta_\ell$ for some $\ell \in J$.
   Then the weight splitting does not change.
   Thus, for any node $i$ with nonzero $a_{i}^{0},b_{i}^{0}$,
   the update corresponds to
   the step~4-2 or 4-2$'$.

    Consider a node $i$ with $a_{i}^{0} = b_{i}^{0} = {\bf 0}$.
    Let $\varLambda$
    be the set of indices $k$
    that attain $\max_{k \in [n]: (a_i)_k \neq 0} \alpha_k$,
    and let $\varPi$ be the set of indices $\ell$
    that attain $\max_{\ell \in [n]: (b_i)_\ell \neq 0} \beta_\ell = - c_i^2$.
    \begin{description}
    \item[{\rm Case 1:}] $\varPi \cap J \neq \emptyset$.
    Then $c_i^2$ does not change and so does $c_i^1$.
    If $\varLambda \cap I \neq \emptyset$, then
    $\kappa$ can increase until $c_i^1$ becomes $- \alpha_k$ for some
$k \in \varLambda$.
    \item[{\rm Case 2:}] $\varPi \cap J = \emptyset$
    ($\Leftrightarrow \varPi \subseteq [n] \setminus J$)
        Then $c_i^2$ changes as $c_i^2 \leftarrow c_i^2 + \kappa$,
        and hence $c_i^1$ changes as $c_i^1 \leftarrow c_i^1 - \kappa$.
        Here $\kappa$ can increase until $\varPi \cap J \neq \emptyset$;
        then the situation goes to (Case 1).
\end{description}
Notice that arc $i \ell$ exists in ${\bar G}_X$ precisely
when $- c_i^2 = \beta_\ell = - c_{\ell}^2$ for $\ell \in X \cap J$,
and hence a node $i$ in the case 2 is precisely a node in $\bar R'$.
Therefore the changes of the weight splitting are the same in {\bf Deg-Det-WMI}
and in the modified Frank's algorithm (using step~4-2$'$).
The steps are iterated for the same zero submatrix until $R^0$ changes.
Therefore,
the total sum of $\kappa$ is the same as that of $\epsilon$ in the
modified Frank's algorithm.
\end{proof}

By this property, the obtained sequences of common independent sets $X$
 can be the same in {\bf Deg-Det-WMI} and the modified Frank's algorithm.
Therefore {\bf Deg-Det-WMI} can also be viewed as yet another implementation
of Frank's algorithm for linear matroids.
A notable feature of {\bf Deg-Det-WMI} is
to skip unnecessary eliminations in constructing the residual graphs.
To see this fact,
consider the partition $\{\sigma_1,\sigma_2,\ldots,\sigma_{n'}\}$ of $[n]$
such that $k,k' \in [n]$ belong to the same part if and only if
$\alpha_k = \alpha_k'$.
Then the elimination matrix $K$
is a block diagonal matrix with block diagonals of size $|\sigma_{i}|
\times |\sigma_{i}|$;
recall (\ref{eqn:Kt^alpha}).
This means that the Gaussian elimination for $A$ in step~1
is done in $O(m \sum_{i}|\sigma_i|^2)$ time.
Therefore, if values $\alpha_k, \beta_\ell$ are scattered,
then $K,L$ are very sparse, and the update of $G_{X}^0$ after $X$
changes is very fast.
On the other hand, necessary eliminations skipped at this moment
will be done in the occurrences of (a-1) and (b-1).
Hence, {\bf Deg-Det-WMI} reduces eliminations compared with
the usual implementation of Frank's algorithm to linear matroids.
More thorough analysis
(e.g.,
incorporating Cunningham's estimate~\cite{Cunningham} for the length of augmenting paths)
is left to a future work.

We close this paper by giving an example in which
the elimination results are actually different in the two algorithms.

\begin{Ex}
Consider matrices
\[
 A = \left(
    \begin{array}{ccccc}
0 & 0 & 1 & 0 & 1 \\
   1 & 1 & 0 & 0 & 0 \\
   0 & 0 & -1 & 1 & 1 \\
0 & 1 & 0 & 1 & 0
    \end{array}
  \right),
 B = \left(
    \begin{array}{ccccc}
1 & 1 & 1 & 0 & 1 \\
   0 & 0 & 0 & 1 & 1 \\
   0 & 0 & 0 & 0 & -1 \\
0 & 1 & 0 & 1 & 1
    \end{array}
  \right)
\]
and weight $c=(3\ 2\ 3\ 1\ 1)$.
The both algorithms for this input can reach at
$\alpha =(-2\ -2\ -2\ -2)$, $\beta = (-1\ 0\ 0\ 0)$
and $X=\{1,2\}$ without elimination.
Consider {\bf Deg-Det-WMI} from this moment.
The matrices $A^0$ and $B^0$ are given by
\[
  A^{0} = \left(
    \begin{array}{ccccc}
0 & 0 & 1 & 0 & 0 \\
   1 & 1 & 0 & 0 & 0 \\
   0 & 0 & -1 & 0 & 0 \\
0 & 1 & 0 & 0 & 0
    \end{array}
  \right),
 B^{0} = \left(
    \begin{array}{ccccc}
1 & 0 & 1 & 0 & 0 \\
   0 & 0 & 0 & 0 & 0 \\
   0 & 0 & 0 & 0 & 0 \\
0 & 1 & 0 & 0 & 0
    \end{array}
  \right).
\]
The Gaussian elimination makes $(a_{2}^{0})_2$ zero.
Then $G_X^0$ consists of one arc $31$, and
$S_X^0 = \{3\}$ and $T_X^0 = \emptyset$.
The reachable set $R^{0}$ is determined as $R^0= \{1,3\}$, and $I,J$
are given by $I = \{1,2,3\}$, $J =\{2,3,4\}$, $I^* = \{1,3\}$, and
$J^* =\{2,3\}$.
Then $\alpha,\beta$ are changed as
$\alpha = (-1\ -1\ -1\ -2)$, $\beta = (-2\ 0\ 0\ 0)$ without
occurrences of (a-1) and (b-1).
Nonzero elements appear in $A^0[I^*,\{4,5\}]$
and $B^0[J^*,\{4,5\}]$, which implies $S^{0}_{X} \cap T_{X}^{0}=\{4,5\}$.
So $X$ is increased.

Therefore {\bf Deg-Det-WMI} succeeds the augmentation
without eliminating $(b_2)_1$,
whereas Frank's algorithm eliminates this element
in constructing $G_X$.
\end{Ex}

\section*{Acknowledgments}
The authors thank Kazuo Murota for comments and the anonymous referee
for careful reading and helpful comments.
The second author was supported by JSPS KAKENHI Grant Numbers JP17K00029
and JST PRESTO Grant Number JPMJPR192A, Japan.


\end{document}